%% file: gk_arxiv.tex
\documentclass[final]{amsart}

\usepackage[utf8]{inputenc}

\input{def2.tex}

\title[Non-existence]{Nonexistence of steady solutions for rotational slender fibre spinning with surface tension}
\author{Thomas Götz}
\address{Thomas~Götz, Mathematical Institute, University Koblenz, Germany.\newline
E--mail:~goetz@uni-koblenz.de}
\author{Axel Klar}
\address{Axel Klar, Dept.~of Mathematics, TU Kaiserslautern, Germany.\newline
E--mail:~klar@mathematik.uni-kl.de}

\begin{document}

\maketitle

\begin{abstract}
Reduced one-dimensional equations for the stationary, isothermal rotational spinning process of slender fibers are considered for the case of large Reynolds ($\delta=3/\re\ll 1$) and small Rossby numbers ($\eps\ll 1$). Surface tension is included in the model using the parameter $\kappa=\sqrt{\pi}/(2\we)$ related to the inverse Weber number. The inviscid case $\delta=0$ is discussed as a reference case. For the viscous case $\delta > 0$ numerical simulations indicate, that for a certain parameter range, no physically relevant solution may exist. Transferring properties of the inviscid limit to the viscous case, analytical bounds for the initial viscous stress of the fiber are obtained. A good agreement with the numerical results is found. These bounds give strong evidence, that for $\delta > 3\eps^2 \lb 1- \frac{3}{2}\kappa +\frac{1}{2}\kappa^2\rb$ no physical relevant stationary solution can exist.
\end{abstract}


\section{Introduction}

Being an important technology, fiber spinning has recently gained considerable attention in the mathematical literature, see~\cite{WDKS02, DKW02, Pan06, PWM08, FZ01, DMM00, DewHowWil94, CKK98, GRRST01, Hag98, Lan97, Yar93, MW08}. Describing a slender fluid jet, the model equations covering spinning processes are settled on balance laws for mass, momentum and energy. Usually, due to slenderness of the fiber, cross--sectional averaging is performed to obtain spatially one--dimensional models~\cite{PWM08, Lan97, DMM00, DewHowWil94, CKK98}. 

In the present work we consider the isothermal, rotational spinning of slender  fibers. As an example of industrial relevance, one might have the production of glass fibers in mind. In this situation energy balance  holds and Coriolis and centrifugal forces as well as surface tension effects have to be included, see~\cite{DKW02, Pan06, MW08}. Neglecting gravity the reduced model equations read in $\R^2$ as
\begin{subequations}
\label{E:Euler1}
\begin{align}
   \pder_t A +\pder_s (Au) &= 0\;, \label{E:Euler1:a}\\
   \pder_t v + u \pder_s v 
	&= \delta \frac{\pder_s (A \pder_s \gamma \pder_s u)}{A} 
	+ \kappa\frac{\pder_s (\sqrt{A}\pder_s\gamma)}{A}
	+ \frac{2}{\eps} v^\perp 
	+ \frac{1}{\eps^2} \gamma\;,\label{E:Euler1:b}\\
   \pder_t \gamma  + u \pder_s \gamma &= v\;, \label{E:Euler1:c}\\
   \norm{\pder_s \gamma} &= 1\;. \label{E:Euler1:d}
\end{align}
\end{subequations}
The independent variables are the arc--length $s \in [0,L]$ and the time $t \in [0, \infty)$. By $\gamma(s,t)\in \R^2$ and $v(s,t)\in \R^2$ we denote the position and the velocity of a point located on the fiber's centerline at time $t$. Additionally, $u(s,t)\in \R$ describes the tangential velocity of a fluid particle moving along the centerline. By $A=A(s,t)$ we denote the cross--sectional area of the fiber. The parameter $\delta=3/\re$ is related to the Reynolds number $\re$ and $\eps$ is the Rossby number, i.e.~the inverse of the scaled rotation frequency. The surface tension parameter $\kappa=\sqrt{\pi}/(2\we)$ is related to the inverse Weber number~\cite{MW08}. Equation~\eqref{E:Euler1:a} is the continuity equation. The momentum equation~\eqref{E:Euler1:b} includes surface tension and viscous, Coriolis and centrifugal forces. Equation~\eqref{E:Euler1:c} describes the fact that the total derivative of $\gamma$ is given by $v$ and finally, Equation~\eqref{E:Euler1:d} is the condition for the arc length parametrization of the fiber which is defined by a normalized tangent.

We mention that equations (\ref{E:Euler1}) have been derived from the full set of three-dimensional Navier-Stokes equations by an asymptotic analysis in the above mentioned papers. It is understood in many papers  that in a boundary layer close to the orifice  the above slender models do not describe the jet well. In this small region near the orifice the flow is fully three-dimensional and other  effects might become important. In the present work we identify parameter regions where the slender model  breaks down and can not be used any more to describe  fiber spinning. For further investigation on the breakup of liquid jets from a rotating orifice we refer to \cite{Par1,Par2}.

In application relevant cases one is typically interested in stationary, i.e.~time independent solutions for small values of the parameters.
The purpose of our present investigations is to identify  those parameter regimes for $\delta$, $\eps$ and $\kappa$ where physically relevant stationary solutions exist for  equations (\ref{E:Euler1})  . In a previous work~\cite{GKUW08} we analyzed the case without surface tension, i.e.~$\kappa=0$. Applying similar techniques, we are able to extend these results to the model including surface tension. The results will show, that surface tension does not help to overcome the problem of non--existence of physically relevant solutions --- a hope that was expressed in the outlook of our earlier work~\cite{GKUW08}.

Finally, we mention that in \cite{AMW10} a one-dimensional rod model allowing for stretching, bending and twisting has been investigated asymptotically and numerically, which allows for simulations in the whole range of parameters. Moreover, a more detailed analysis of the parameter regions where the string model breaks down is given there for the case without surface tension.

This work is organized as follows:~In Section~\ref{S:Simplify}, we rewrite the stationary version of the above system~\eqref{E:Euler1} in a more tractable form. Appropriate boundary conditions are specified. Section~\ref{S:Inviscid} is devoted to the inviscid case $\delta=0$. The main result is contained in Section~\ref{S:Viscous}, where we characterize the stationary solutions of the viscous problem $\delta>0$ and show the non--existence of physically reasonable stationary solution for some range of the parameters $\delta$, $\eps$ and $\kappa$. In Section~\ref{S:Numerics} we present numerical simulations based on a solution of the boundary value problem in Eulerian coordinates. The obtained results confirm the analysis.

\section{The stationary case}
\label{S:Simplify}

In Eqn.~\eqref{E:Euler1:a}, the mass flux is given by $Au$. In the stationary case, the mass flux is constant and after appropriate scaling one gets $Au=1$. Physical meaningful solutions are characterized by strictly positive $u$, hence $A=1/u$. Furthermore,~\eqref{E:Euler1:c} reads as $v=u\pder_s \gamma$ and inserting this into~\eqref{E:Euler1:b} yields
\begin{subequations}
\label{E:Euler2}
\begin{align}
\label{E:Euler2:a}
   \pder_s(u\,\pder_s \gamma) 
	&= \delta \pder_s \lb \frac{\pder_s u\, \pder_s\gamma}{u}\rb 
	+ \kappa \pder_s \lb \frac{\pder_s\gamma}{\sqrt{u}}\rb
	+ \frac{2}{\eps} \pder_s\gamma^\perp 
	+ \frac{1}{\eps^2}\frac{1}{u}\gamma\;, \\
	\norm{\pder_s\gamma} &= 1\;. \label{E:Euler2:b}
\end{align}
\end{subequations}
To eliminate the algebraic constraint~\eqref{E:Euler2:b}, we project $\pder_s \gamma$ onto the unit circle and parameterize $\pder_s \gamma :=\tau = (\cos\alpha, \sin\alpha)$ by the angle $\alpha=\alpha(s)$. Furthermore, we use polar coordinates $\gamma=r(\cos\phi, \sin\phi)$ for the centerline and introduce $\beta=\alpha-\phi$. Abbreviating derivatives with respect
to $s$ by a prime we transform~\eqref{E:Euler2:a} via 
$q=u-\delta u'/u-\kappa/\sqrt{u}$ into a first order system which we supply with boundary conditions at $s=0$ (nozzle) and $s=L$ (end point of the fiber):
\begin{subequations}
\label{E:Euler3}
\begin{align}
   \delta u' &= u \lb u-\frac{\kappa}{\sqrt{u}} - q\rb\;, 
	& u(0) &= 1\;,\\
   \eps^2 u q' &= r\cos \beta\;, 
	& q(L) &= u(L)-\frac{2\kappa}{\sqrt{u(L)}}\;,\\
   r' &= \cos\beta\;, & r(0) &= 1\;,\\
   \beta' &= -\frac{2}{\eps q} 
	- \lb \frac{r^2}{\eps^2 u q} +1\rb \frac{\sin\beta}{r}\;, 
	& \beta(0) &= 0\;.
\end{align}
\end{subequations}

\begin{rem}
The above boundary conditions for $u$, $r$ and $\beta$ given at $s=0$ express the fact that the fiber leaves the nozzle tangentially with a fixed velocity. Due to the boundary condition for $q$ at the end of the fiber, i.e.~at $s=L$, the viscous stresses at the end of the fiber balance the surface tension. Note, that in the inviscid case $\delta=0$ without surface tension, i.e.~$\kappa=0$, we have due to positive $u$ the equation $u(s)=q(s)$ for all $s$.
\end{rem}

\begin{lem}
\label{L:dq_pos}
In~\eqref{E:Euler3} it holds that $q'(0)>0$.
\end{lem}

\begin{rem}
The observable $q$ can be interpreted as the internal energy of the fluid being the sum of the kinetic energy $Au^2=u$ and the viscous stress energy $-\delta Au'=-\delta u'/u$ and the potential energy $-\kappa \sqrt{A}=-\kappa/\sqrt{u}$ due to surface tension.
\end{rem}

\begin{rem}
\label{R:dphi_null}
The polar angle $\phi$ of the centerline is determined via
\begin{align}
   \phi' = \frac{\sin\beta}{r}\;, \qquad \phi(0) = 0
\end{align}
leading to $\phi'(0)=0$.
\end{rem}

\begin{rem}[Equations in Lagrangian coordinates]
For describing the parameter regimes for $\delta$, $\eps$ and $\kappa$ where physically relevant solutions of \eqref{E:Euler3} exist, it is convenient to introduce Lagrangian coordinates. The passage from Eulerian to Lagrangian coordinates is settled on the transformation
\begin{equation*}
   \frac{ds(t)}{dt} = u(s(t))>0, \qquad s(0) = 0
\end{equation*}
mapping the Euler position $s$ of a material point to the Lagrangian parameter $t$. Starting from the nozzle it takes a fluid particle $t$ time units to approach position $\gamma(s)$. It takes a particle $T$ time units to travel from the nozzle to the endpoint $s=L$.

After some elementary manipulations we obtain the Lagrangian system
\begin{subequations}
\label{E:Lagrange1}
\begin{align}
   \delta \dot{u}_L 
	&= u_L^2 \lb u_L-\frac{\kappa}{\sqrt{u_L}}-q_L\rb\;, 
		& u_L(0) &= 1\;,\\
   \eps^2 \dot{q}_L &= r_L\cos \beta_L\;, 
		& q_L(T) &= u_L(T)-\frac{2\kappa}{\sqrt{u_L(T)}}\;,\\
   \dot{r}_L &= u_L \cos\beta_L\;, & r_L(0) &= 1\;,\\
   \dot{\beta}_L &= -\frac{2 u_L}{\eps q_L} 
	- \lb \frac{r_L^2}{\eps^2 q_L} +u_L\rb \frac{\sin\beta_L}{r_L}\;,
	& \beta_L(0) &= 0\;.
\end{align}
\end{subequations}
Let $u_L(t)=u(s(t))$ be the velocity in the Lagrangian framework and mutatis mutandis for the other variables. By a dot, e.g.
\begin{equation*}
   \dot{u}_L=\frac{du_L(t)}{dt} 
	= \frac{d u(s(t))}{ds}\frac{ds(t)}{dt} = u' u\;,
\end{equation*}
we denote the derivative with respect to the Lagrangian parameter.
\end{rem}

\section{Inviscid limit $\delta=0$}
\label{S:Inviscid}

As outlined in the Introduction, we are interested in cases where $\delta$, $\eps$ and $\kappa$ are small. In this Section we shall exploit this fact in a formal manner. As small $\eps$ corresponds to a high rotation frequency one expects $u=O(1/\eps)$ at least in the fiber's interior. Due to the definition of $q$ one also expects $q=O(1/\eps)$. Therefore it is convenient to introduce the re--scaled variables
\begin{equation*}
   v = \eps u\; , \qquad w=\eps q\; .
\end{equation*}
In terms of $v,w,r$ and $\beta$ the model equations~\eqref{E:Euler3} become
\begin{subequations}
\label{E:Euler3s}
\begin{align}
   \delta\eps v' &= v \lb v-\frac{\lambda}{\sqrt{v}}- w\rb\;,
	& v(0) &=\eps\;,\\
   v w' &= r\cos \beta\;, 
	& v(L) &= w(L)+\frac{\lambda}{\sqrt{v(L)}}\;,\\
   r' &= \cos\beta\;, & r(0) &= 1\;,\\
   \beta' &= -\frac{2}{w} 
	- \lb \frac{r^2}{vw} +1\rb \frac{\sin\beta}{r}\;, 
	& \beta(0) &= 0\;,
\end{align}
\end{subequations}
where $\lambda=\eps^{3/2}\kappa$. In the ---formal--- inviscid limit $\delta=0$ the system~\eqref{E:Euler3s} together with the assumed positivity of $v$ yields

\begin{lem}
\label{L:ugleichq}
In the inviscid limit $\delta=0$ of the stationary system~\eqref{E:Euler3s} it holds that
\begin{equation*}
   w=v-\frac{\lambda}{\sqrt{v}} \quad \text{i.e.} \quad
	 q = u - \frac{\kappa}{\sqrt{u}}\;.
\end{equation*}
\end{lem}

In the inviscid limit $\delta=0$, we can eliminate $w$ in the above system~\eqref{E:Euler3s} and obtain the reduced system
\begin{subequations}
\label{E:Euler:delta0}
\begin{align}
   v'	&= \lb v +\frac{\lambda}{2\sqrt{v}}\rb^{-1} 
		r\, \cos\beta\;,
	& v(0) &= \eps\;,\\
   r' &= \cos\beta\;, & r(0) &= 1\;, \\
   \beta' &= - \frac{2}{v-\frac{\lambda}{\sqrt{v}}} - 
	\lb \frac{r^2}{v^2-\lambda \sqrt{v}}+1\rb
		\frac{\sin\beta}{r}\;, 
	& \beta(0) &= 0\;.
\end{align}
\end{subequations}
This system can be solved numerically. The Figures~\ref{F:rb_invis} and ~\ref{F:sv_invis} show simulations for $\eps=0.16$ and $\kappa=0.5$ leading to $\lambda=0.032$ compared to the results without surface tension, i.e.~$\kappa=\lambda=0$.
\begin{figure}
\begin{minipage}[t]{.45\textwidth}
\vspace*{0pt}\par
\includegraphics[width=\textwidth]{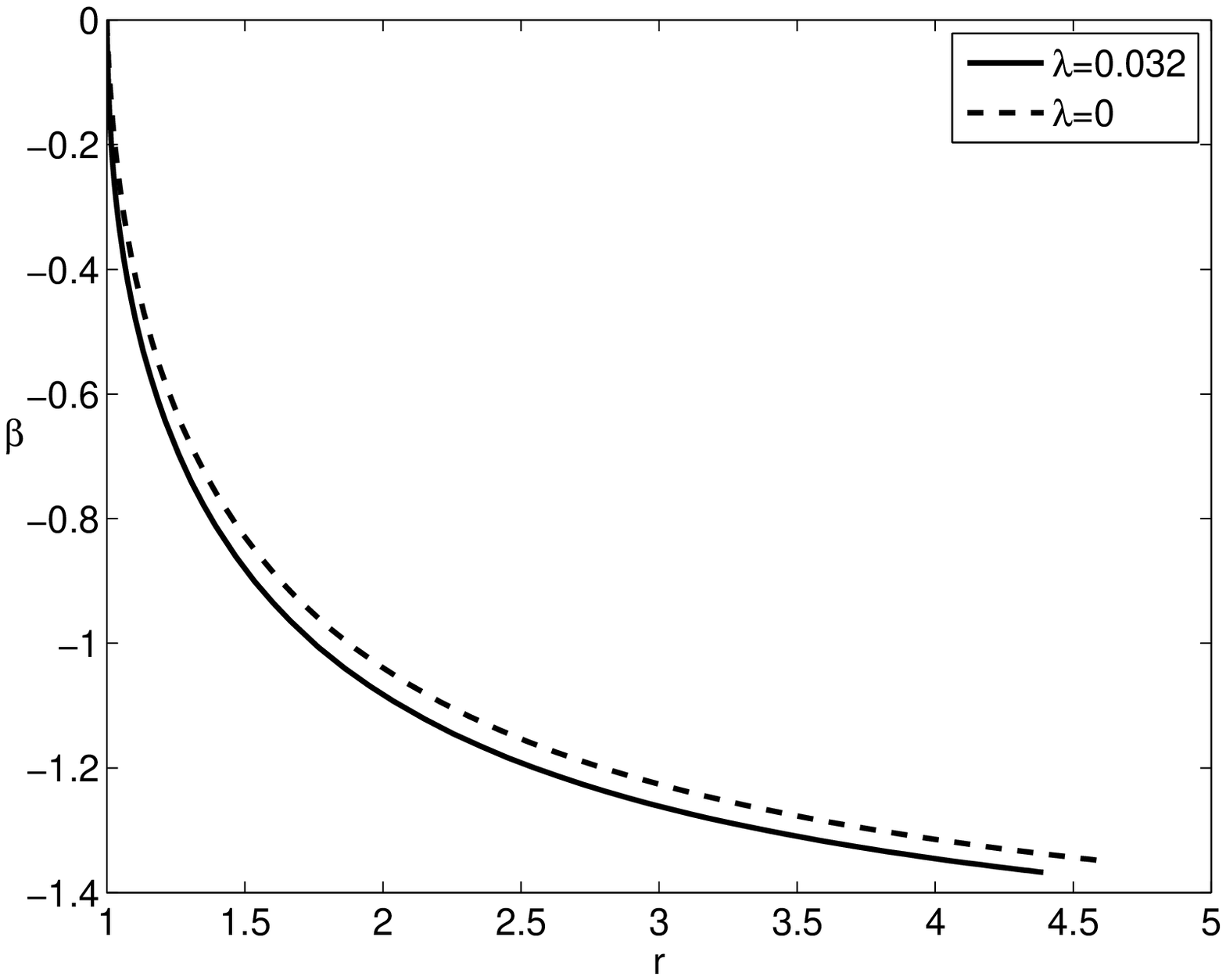}
\caption{\label{F:rb_invis} Phaseportrait of the system~\eqref{E:Euler:delta0} in the $r\beta$--plane.}
\end{minipage}
\hfill
\begin{minipage}[t]{.45\textwidth}
\vspace*{0pt}\par
\includegraphics[width=\textwidth]{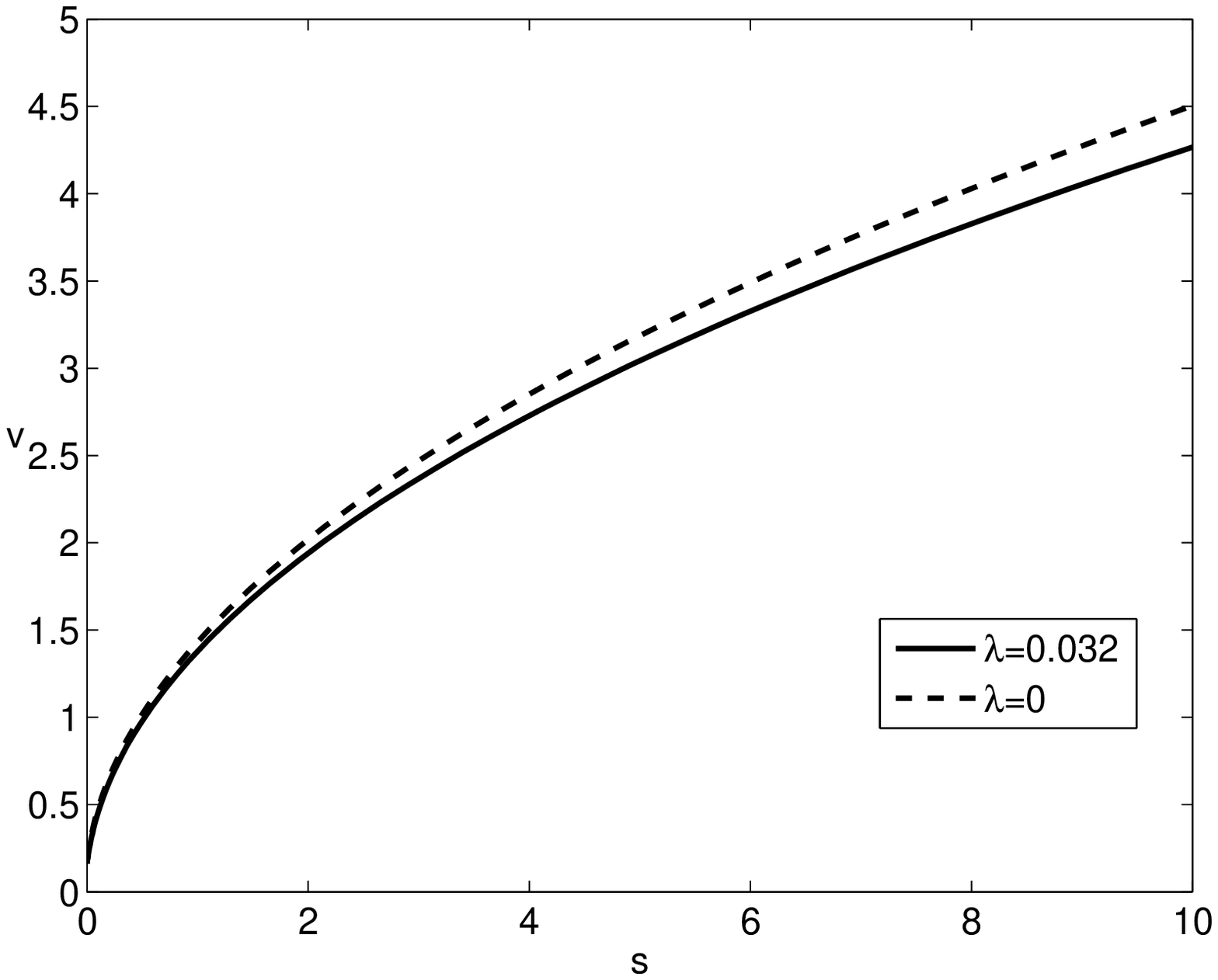}
\caption{\label{F:sv_invis} Velocity $v=\eps u$ of an inviscid fiber with and without surface tension.}
\end{minipage}
\end{figure}
The numerical results show, that including surface tension to the inviscid model alters the behaviour of the fiber only marginally.

\begin{lem}
\label{L:conc}
For $0<\eps<\dfrac{\sqrt{1-\kappa/4}}{1+\kappa/2}$ it holds that
\begin{equation}
   u''(0) = \frac{\eps^2\lb 1+\frac{\kappa}{2}\rb^2+\frac{\kappa}{4}-1}{\eps^4\lb 1+\frac{\kappa}{2}\rb^3}<0\;.
\end{equation}
\end{lem}
\begin{proof}
Considering the system~\eqref{E:Euler:delta0}, the expression for $v''$ reads as
\begin{align*}
   v'' &= \frac{r'\cos \beta - r\beta'\sin\beta}{v+\frac{\lambda}{2\sqrt{v}}} - r v'\cos\beta \frac{1-\frac{\lambda}{4\sqrt{v}^3}}{\lb v+\frac{\lambda}{2\sqrt{v}}\rb^2}
\intertext{and using the boundary conditions at $s=0$, we obtain}
   v''(0) &= \frac{\eps^2\lb 1+\frac{\kappa}{2}\rb^2+\frac{\kappa}{4}-1}{\eps^3\lb 1+\frac{\kappa}{2}\rb^3}\;.
\end{align*}
\end{proof}

Lemma~\ref{L:conc} shows, that in the inviscid, fast rotating limit the velocity in Eulerian coordinates is locally at the nozzle a concave function.

\begin{lem}\label{L:conv}
For the second derivative of the Lagrangian velocity $u_L$ it holds that
\begin{equation}
   \ddot{u}_L(0) >0\;.
\end{equation}
\end{lem}
\begin{proof}
Using~\eqref{E:Lagrange1} with $\delta =0$, we get
\begin{equation*}
   \eps^2 \dot{u}_L = r_L \cos\beta_L 
		\lb 1+\frac{\kappa}{2u_L^{3/2}}\rb^{-1}\;.
\end{equation*}
Differentiating again and inserting the conditions at $s=0$ yields
\begin{equation*}
   \ddot{u}_L(0) = \frac{1}{\lb 1+\frac{\kappa}{2}\rb^3}
	\lsb \frac{\lb 1 +\frac{\kappa}{2}\rb^2}{\eps^2} + \frac{3\kappa}{4\eps^4}\rsb >0\;.
\end{equation*}
\end{proof}

Lemma~\ref{L:conv} shows, that in the Lagrangian framework of the inviscid limit $\delta=0$, the velocity is a convex function at least close to the nozzle. Note, that $\ddot{u}> 0$ implies that the acceleration of a fluid particle increases with its flight time. This agrees with the fact, that the acting forces, i.e.~the centrifugal and Coriolis forces, also increase the further outward the fluid particle travels.

\section{The viscous case}
\label{S:Viscous}

After the investigation of the inviscid case $\delta=0$, we turn our attention to the boundary value problem (BVP)~\eqref{E:Euler3} governing the stationary viscous case
\begin{subequations}
\label{E:Euler4}
\begin{align}
   \delta u' &= u\lb u-\frac{\kappa}{\sqrt{u}}-q\rb\;, 
	& u(0) &= 1\;, \label{E:Euler4:a}\\
   \eps^2 u q' &= r\cos\beta\;, 
	& q(L) &= u(L)-\frac{2\kappa}{\sqrt{u(L)}}\;, \label{E:Euler4:b}\\
   r' &= \cos \beta\;, & r(0) &= 1\;, \label{E:Euler4:c}\\
   q \beta' &= -\frac{2}{\eps} - \lb \frac{r^2}{\eps^2 u}+q\rb \frac{\sin\beta}{r}\;, & \beta(0) &= 0\;. \label{E:Euler4:d}
\end{align}
\end{subequations}
To gain some insight into the behaviour of the solutions for small values of the parameters $\delta$, $\eps$ and $\kappa$, we carry out numerical simulations. The results are obtained with \emph{Matlab$^\circledR$ 2008a} using the routine \texttt{bvp5c}, a finite difference code that implements the four--stage Lobatto IIIa formula. This is a collocation formula and the collocation polynomial provides a $C^1$--continuous solution that is fifth--order accurate uniformly on the considered interval. The formula is implemented as an implicit Runge--Kutta formula and solves the algebraic equations directly. As an initial guess for the solution we use the inviscid case $\delta=0$. 

\begin{table}[h]
\begin{center}
\begin{tabular}{llll|llll}
\multicolumn{4}{c|}{$\eps=0.25, \kappa=0.1$}
&\multicolumn{4}{c}{$\eps=0.1, \delta=0.01$} \\
$\delta$	& $q(0)$	& $u''(0)$	& $\ddot{u}(0)$ 	&
$\kappa$	& $q(0)$	& $u''(0)$	& $\ddot{u}(0)$ \\
\hline
0.1	& 0.1252 	& -4.8225	& 55.221	&
0.3	& 0.1535	& -208.43	& 2797.9 \\	
0.125	& 0.0239	& -3.7310	& 45.307	&
0.4	& 0.0653	& -194.37	& 2681.0 \\
0.13	& 0.0092	& -3.5535 	& 43.348	&
0.45	& 0.0220	& -168.59	& 2617.5 \\
0.133	& 0.0021	& -2.2389	& 43.288	&
0.475	& 0.0013		& -149.47	& 2592.1 \\
0.135	& \multicolumn{3}{c|}{\emph{no convergence}} &
0.48	& \multicolumn{3}{c}{\emph{no convergence}}\\[1ex]
\end{tabular}
\caption{\label{T:Results} Results of the numerical simulations using Matlab$^\circledR$.}
\end{center}
\end{table}

The numerical results, some of them are given in Table~\ref{T:Results}, as well as the analytical results for the inviscid case (see Lemmata~\ref{L:ugleichq},~\ref{L:conc} and~\ref{L:conv}) show, that a  solution of the stationary system~\eqref{E:Euler4}, if it exists at all, has the following properties:
\begin{defi}[Physically relevant stationary solutions]
\label{D:PRS}
We call a solution $(u,q,r,\beta)$ of the boundary value problem~\eqref{E:Euler4} a \emph{physically relevant stationary solution}, if it has the following properties
\begin{align}
   q_0 :=q(0) & \in (0,1-\kappa]\;, \label{P:q0} \tag{P1} \\
   u''(0) &< 0\;, \label{P:u_concave} \tag{P2} \\
   \ddot{u}_L(0) &> 0\;. \label{P:uL_convex} \tag{P3}
\end{align}
\end{defi}

\begin{rem}In the sense of this definition, all the numerical solutions we were able to compute are physically relevant; the properties~\eqref{P:q0},~\eqref{P:u_concave} and~\eqref{P:uL_convex} are satiesfied. Furthermore, the stationary solution of the inviscid case, i.e.~the solution of~\eqref{E:Euler:delta0} together with the according solutions for $u$ and $q$ is physically relevant.
\end{rem}

Let us discuss the properties~\eqref{P:q0}---\eqref{P:uL_convex} in detail.

\begin{lem}
\label{L:equivcond}
Let $(u,q,r,\beta)$ be a solution of the BVP~\eqref{E:Euler4}.
Then, property~\eqref{P:q0} is equivalent to $u'(0)\ge 0$ and $\beta'(0) < 0$.
\end{lem}
\begin{proof}Eqn.~\eqref{E:Euler4:a} reads as
\begin{align*}
   u'(0) = \frac{u}{\delta} \lb u-\frac{\kappa}{\sqrt{u}}- q\rb\Big|_0 
	= \frac{1-\kappa-q_0}{\delta}
\intertext{and hence}
   u'(0)\ge 0 \quad \iff q_0 & \le 1-\kappa\;.
\end{align*}
Due to~\eqref{E:Euler4:d}
\begin{equation*}
   \beta'(0) = -\frac{2}{\eps q_0}
\end{equation*}
holds and $\beta'(0)<0$ if and only if $q_0>0$.
\end{proof}

\begin{rem}
\label{R:InterpretP1}
The conditions $u'(0)\ge 0$ and $\beta'(0)<0$ allow for an easy interpretation in physical terms:\\
The condition $u'(0)\ge 0$ assures, that the fiber is not accelerated \emph{into} the nozzle---a fact that one should expect from a physically relevant solution.\\
Due to Remark~\ref{R:dphi_null} we observe that the condition $\beta'(0)<0$ is equivalent to  $\alpha'(0)=\beta'(0)-\phi'(0) = \beta'(0)<0$. Hence the fiber initially bends \emph{in} the sense of rotation of the drum. The opposite condition $\alpha'(0)>0$ would imply, that the fiber initially bends \emph{against} the direction of rotation of the drum. This is obviously counterintuitive and \emph{physically not reasonable}.
\end{rem}

\begin{rem}
The property~\eqref{P:q0} also implies, that physically relevant solutions require $0\le \kappa<1$.
\end{rem}

\begin{lem}
Let $(u,q,r,\beta)$ be a stationary solution of the BVP~\eqref{E:Euler4}. Then, the property~\eqref{P:q0} implies
\begin{equation*}
   \delta^2 u''(0) < \lb 2-\frac{\kappa}{2}\rb \lb 1-\kappa\rb - \frac{\delta}{\eps^2}\;.
\end{equation*}
In other words: For all $\delta>0$ and $\eps^2<\dfrac{\delta}{\lb 2-\frac{\kappa}{2}\rb \lb 1-\kappa\rb}$: \eqref{P:q0} $\Longrightarrow$~\eqref{P:u_concave}.
\end{lem}
\begin{proof}
Via differentiation we deduce from~\eqref{E:Euler4} the equation
\begin{equation*}
   \delta^2 u'' = u'\lb u-\frac{\kappa}{\sqrt{u}}-q\rb 
	+ u \lb u'+\frac{\kappa u'}{2\sqrt{u}^3}-q'\rb 
	= \frac{1}{\delta} \lb 2u-\frac{\kappa}{2\sqrt{u}}-q\rb
	\lb u-\frac{\kappa}{\sqrt{u}}-q\rb - \frac{r\cos\beta}{\eps^2}
\;.
\end{equation*}
As a consequence we obtain at the nozzle, i.e.~for $s=0$
\begin{equation}
   \label{E:aux2}
   \delta^2 u''(0) = \lb 2-\frac{\kappa}{2}-q_0\rb 
	\lb 1-\kappa-q_0\rb - \frac{\delta}{\eps^2}\;,
\end{equation}
and the property~\eqref{P:q0} yields
\begin{equation*}
   \delta^2 u''(0) < \lb 2-\frac{\kappa}{2}\rb \lb 1-\kappa\rb - \frac{\delta}{\eps^2}\;.
\end{equation*}
\end{proof}

This results shows, that $u''(0) < 0$ for all $\delta$ and $\eps$ with $2\eps^2 <\dfrac{\delta}{\lb 2-\frac{\kappa}{2}\rb\lb 1-\kappa\rb}$. As this statement holds for fixed $\delta$ for all ``small'' parameters $\eps$, we can deduce~\eqref{P:u_concave}  rigorously from~\eqref{P:q0} in the case of small $\eps$ and small $\kappa$. In particular, for fixed $\delta$, the property~\eqref{P:u_concave} does not result in a restiction on the values of $\eps$ and $\kappa$ as $\eps$ and $\kappa$ tend to $0$.

The property~\eqref{P:uL_convex} of a physically relevant solution is required due to numerical evidence and analogy to the inviscid case,. See also the arguments at the end of section \ref{S:Inviscid}.
It will lead us now to bounds for $q_0$ and finally to conditions on $\kappa$, $\eps$ and $\delta$. 

\begin{thm} \label{T:Bound}
Let $(u,q,r,\beta)$ be a physically relevant stationary solution of the BVP~\eqref{E:Euler4}. Then, the following estimate for the initial value $q_0$ holds:
\begin{equation*}
   \max \lb 0, \underline{q_0}\rb \le q_0 
	\le \min \lb 1, \overline{q_0}\rb
\end{equation*}
where
\begin{align*}
   \underline{q_0} &= \frac{3-\frac{3\kappa}{2}-\sqrt{\lb 1+\frac{\kappa}{2}\rb^2+4\delta\eps^{-2}}}{2}\;, \\
   \overline{q_0} &= \frac{5-\frac{7\kappa}{2}-\sqrt{\lb 1+\frac{\kappa}{2}\rb^2+8\delta\eps^{-2}}}{4}\;.
\end{align*}
\end{thm}
\begin{proof}
Due to~\eqref{E:aux2} we have
\begin{equation*}
   \delta^2 u''(0) = \lb 2-\dfrac{\kappa}{2}-q_0\rb \lb 1- \kappa - q_0\rb  - \frac{\delta}{\eps^2} = \mu\;.
\end{equation*}
The property~\eqref{P:u_concave} requires $\mu<0$ and hence we obtain
\begin{equation*}
   \underline{q_0} = \frac{3-\frac{3\kappa}{2}-\sqrt{\lb 1+\frac{\kappa}{2}\rb^2+4\delta\eps^{-2}}}{2}
	< q_0 
	< \frac{3-\frac{3\kappa}{2}+\sqrt{\lb 1+\frac{\kappa}{2}\rb^2+4\delta\eps^{-2}}}{2} \ge 1\;.
\end{equation*}
Analogously, due to $\ddot{u}_L(0) = u''(0)u(0)^2+u'(0)^2 u(0)$ and using~\eqref{E:aux2} again, we get
\begin{equation*}
   \delta^2 \ddot{u}_L(0) = \lb 1-\kappa -q_0\rb \, \lb 3-\frac{3}{2}\kappa-2q_0\rb -\frac{\delta}{\eps^2}
	= \lambda\;.
\end{equation*}
Now,~\eqref{P:uL_convex}, i.e.~$\lambda\ge 0$ leads to
\begin{equation*}
   q_0 \ge \frac{5-\frac{7\kappa}{2}+\sqrt{\lb 1+\frac{\kappa}{2}\rb^2+8\delta\eps^{-2}}}{4}>1
	 \quad\text{or}\quad
   q_0 \le \overline{q_0} 
	= \frac{5-\frac{7\kappa}{2}-\sqrt{\lb 1+\frac{\kappa}{2}\rb^2+8\delta\eps^{-2}}}{4}\;.
\end{equation*}
Summarizing the individual estimates shows the assertion.
\end{proof}

\begin{cor}[Non--existence of physically relevant solutions]
\label{C:nonexistence}
If 
\begin{equation}
\label{E:nonexistence}
  p (\kappa) =  3\lb 1-\frac{3\kappa}{2}+\frac{\kappa^2}{2}\rb \le \frac{\delta}{\eps^2}
\end{equation}
no physically relevant stationary solution exists.
\end{cor}
\begin{proof}
If $\overline{q_0}=0$, physically relevant stationary solutions cannot exist.
\end{proof}

\bigskip

\begin{rem}
The results reported in Table~\ref{T:Results2} confirm these findings
\begin{table}[h]
\begin{center}
\begin{tabular}{lll|llr}
$\delta$ & $\eps$ & $\kappa$ & $\delta/\eps^2$ & $p(\kappa)$ & \\
\hline
0.133	& 0.25 	& 0.1		& 2.128	&  2.565 & convergence\\
0.135	& 0.25	& 0.1		& 2.16	& 2.565 & \emph{no convergence} \\[.5ex]
0.01		& 0.1		& 0.475	& 1		& 1.201 & convergence \\
0.01		& 0.1		& 0.48	& 1		& 1.186 & \emph{no convergence}
\end{tabular}
\caption{\label{T:Results2} Numerical results for different values of the parameters}
\end{center}\end{table}
The boundary between numerical convergence and not--convergence is close to the analytical non--existence result~\eqref{E:nonexistence}.
\end{rem}
\begin{rem}
The result of the corollary can be partly understood, if the parameters are reinterpreted in their physical context. We have, that $\delta=3/\re$ is proportional to the kinematic viscosity $\nu$. Furthermore $\eps$, being the Rossby--number, is proportional to the inverse of the rotation frequency $\omega$ of the spinning drum and $\kappa$ is proportional to the surface tension $\sigma$. With this interpretation in mind, Corollary~\ref{C:nonexistence} states, that with some scaling constant $c>0$ and with  $\tilde p (\sigma) = p(\kappa)  $, we have that 
\begin{itemize}
\item for $\nu < \overline{\nu}:= \dfrac{c}{\omega^2} \tilde p (\sigma)$ physically relevant stationary solutions may exist and
\item for $\nu > \overline{\nu}$ physically relevant stationary solutions cannot exist.
\end{itemize}
Hence, for a given rotation frequency $\omega$ and surface tension $\sigma$, fibers with a viscosity $\nu>\overline{\nu}$ cannot be spun.
The larger the surface tension, the smaller the limiting viscosity $\bar \nu$ has to be. Moreover, there is a limit surface tension $\bar \sigma$ related to the value $\kappa =1$. In case $\sigma  >  \bar \sigma$ there is no solution, independently of the values of $\epsilon$ and $\delta $.
\end{rem}

\section{Numerical results}
\label{S:Numerics}

The Figures~\ref{F:rb_invis} and~\ref{F:sv_invis} on page~\pageref{F:rb_invis} show simulations for $\eps=0.16$ and $\kappa=0.5$ compared to the results without surface tension, i.e.~$\kappa=0$. These simulations indicate that including surface tension alters the inviscid case just marginally.

Figure~\ref{F:simures} shows the parameter values, for which the equations~\eqref{E:Euler4} were solved successfully. Table~\ref{T:Results} reports some of these results. The solid line corresponds to the theoretical bound (\ref{E:nonexistence}). In Corollary~\ref{C:nonexistence} it was shown that only in the region below the solid line physically relevant stationary solutions may exist. Our numerical simulations confirm this finding.

Figure~\ref{F:q0_kapfest} illustrates the behavior of $q_0$ for fixed surface tension $\kappa=0.1$. Simulations are shown for $\eps=0.1$ and $\eps=0.24$ as well as for different values of $\delta$. The symbols indicate the numerical results, whereas the lines show the bounds for $q_0$ due to Theorem~\ref{T:Bound}. It is clearly visible that all our numerical results are within the limits predicted due to the assumption of physically relevant solutions.

\begin{figure}[ht]
\begin{minipage}[t]{.45\textwidth}
\vspace*{0pt}\par
\includegraphics[width=\textwidth]{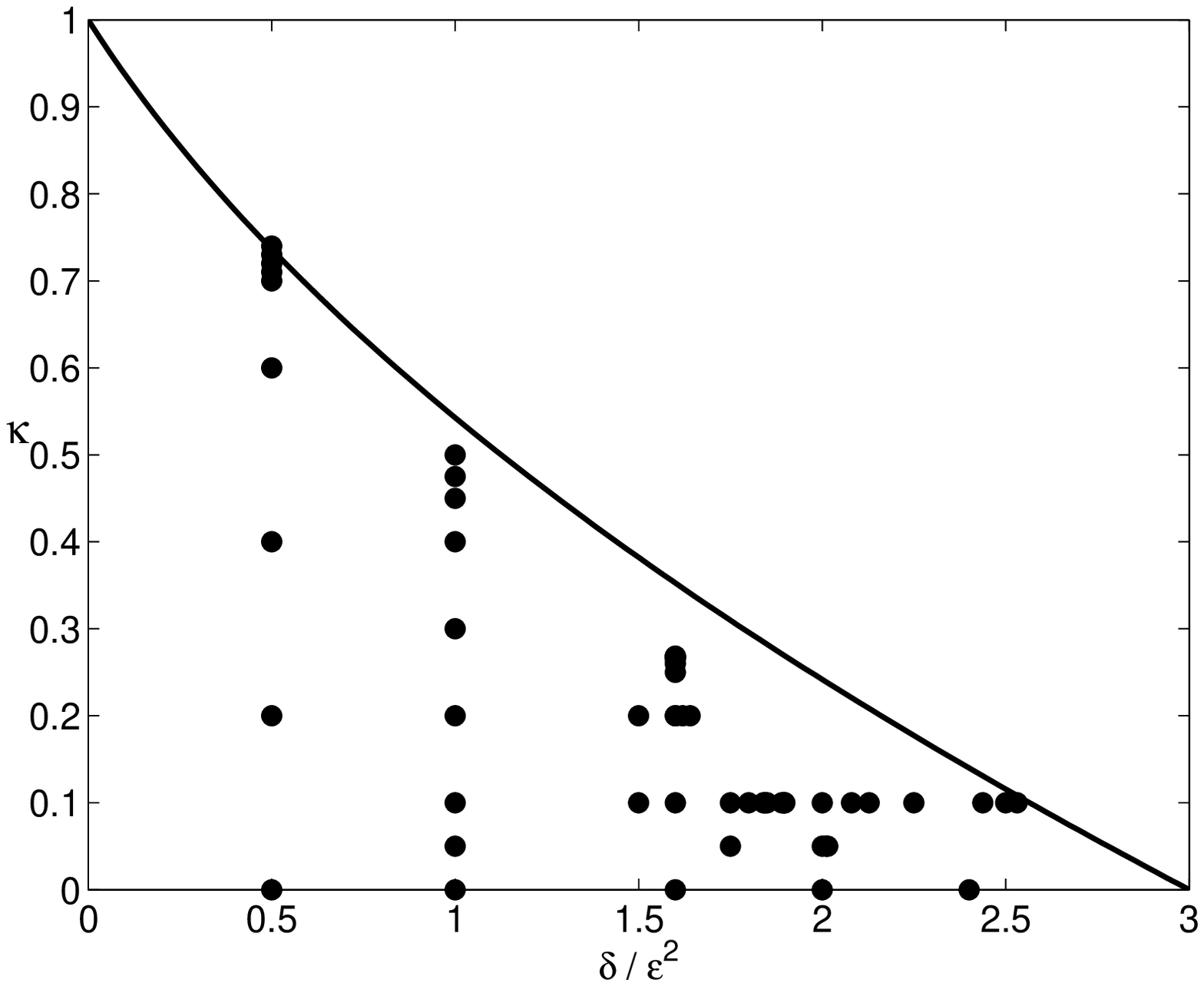}
\caption{\label{F:simures} The parameters $\delta/\eps^2$ and $\kappa$ for the successful simulations indicated by dots. The solid lines shows the theoretical bound~\eqref{E:nonexistence}.}
\end{minipage}
\hfill
\begin{minipage}[t]{.45\textwidth}
\vspace*{0pt}\par
\includegraphics[width=\textwidth]{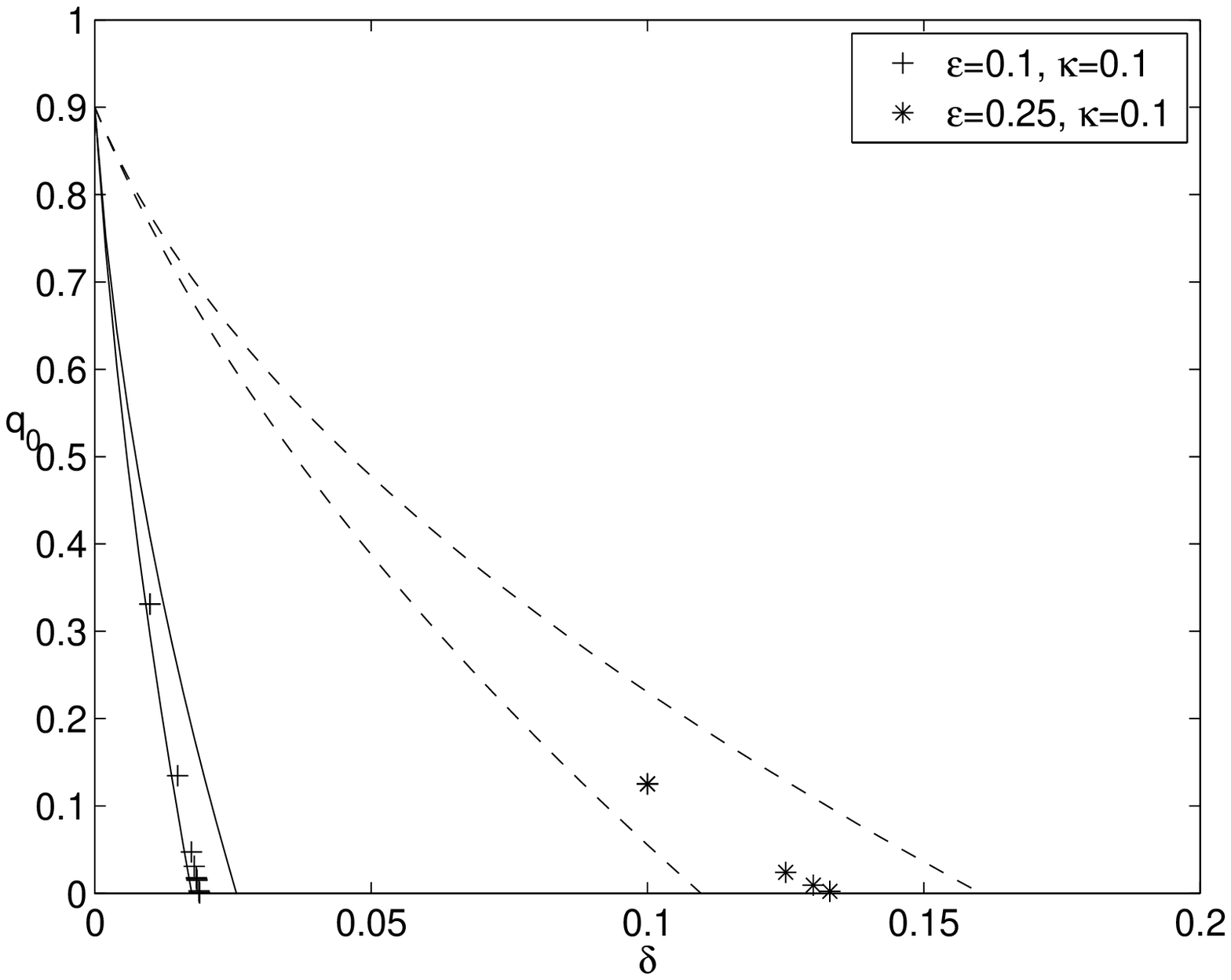}
\caption{\label{F:q0_kapfest} Plot of $q_0$ vs.~$\delta$ for $\kappa=0.1$ and $\eps=0.1,\, 0.25$.}
\end{minipage}
\end{figure}

\section{Conclusions}

There is strong numerical evidence that the regime where physically relevant stationary solutions exist is characterized by the inequality
\begin{equation}
\label{E:glory} 
   \delta < \eps^2 p(\kappa)\;.
\end{equation}

Naturally, one has to interpret a coupling of $\delta, \kappa$ and $\eps$ like~\eqref{E:glory} in terms of the kinematic viscosity $\nu$  of the fluid, the rotational frequency $\omega$ of the drum and  the surface tension $\sigma $. In these variables the statement above is translated into the following:
for  $\nu > \overline{\nu}:= \dfrac{c}{\omega^2} \tilde p (\sigma)$ the fibre can not be spun.

A detailed analysis---both theoretically and numerically--- of the behavior of transient solutions to~\eqref{E:Euler1} may reveal, whether physically relevant \emph{transient} solutions show up in the parameter regime, where stationary solutions cease to exist.

\bibliographystyle{unsrt}
\bibliography{literatur.bib}

\end{document}

%% file: def2.tex
\usepackage{amsmath, amssymb}
\usepackage{graphicx}
\usepackage{dcolumn}

\newcommand{\R}{\mathbb{R}}

\newcommand{\eps}{\varepsilon}
\newcommand{\re}{\text{Re}}

\newcommand{\we}{\text{We}}

\renewcommand{\phi}{\varphi}

\renewcommand{\bar}[1]{\overline{#1}}

\newcommand{\norm}[1]{\left\|#1\right\|}

\newcommand{\pder}{\partial}

\newcommand{\lb}{\left(}
\newcommand{\lsb}{\left[}

\newcommand{\rb}{\right)}
\newcommand{\rsb}{\right]}

\newcommand{\lpm}{\begin{pmatrix}}
\newcommand{\rpm}{\end{pmatrix}}

\theoremstyle{plain}
\newtheorem{thm}{\bf Theorem}[section]
\newtheorem{lem}[thm]{\bf Lemma}

\newtheorem{cor}[thm]{\bf Corollary}
\newtheorem{defi}{\bf Definition}

\theoremstyle{remark}
\newtheorem{rem}[thm]{\bf  Remark}

\allowdisplaybreaks[1]

\numberwithin{equation}{section}
\numberwithin{table}{section}
\numberwithin{figure}{section}

\newcolumntype{d}[1]{D{.}{.}{2.#1}}
\newcolumntype{2}{D{.}{.}{1.2}}
\newcolumntype{5}{D{.}{.}{2.5}}
